\RequirePackage{fix-cm}
\documentclass[smallextended]{svjour3}       
\smartqed  
\usepackage{graphicx}
\usepackage{txfonts}
\usepackage{mathdots}
\usepackage{booktabs}
 \journalname{  }

\usepackage{amssymb}
\usepackage{CJK}
\usepackage{mathrsfs}
\usepackage{epsf,epsfig,amsfonts,amsgen,amstext,amsbsy,amsopn,latexsym}

\newtheorem{cor}[theorem]{Corollary}

\begin{document}
\title{Duality of generalized twisted Reed-Solomon codes and Hermitian
self-dual MDS or NMDS codes\\ \thanks{This research was funded by
the National Natural Science Foundation of China (Grant No.
11901579), Natural Science Foundation of Shaanxi Province (Grant
Nos. 2021JQ-335 and 2021JM-216), and Support fund for Excellent
Doctoral Dissertation of Air Force Engineering University (Grant No.
KGD083920015). } }


\author{Guanmin Guo  \and Ruihu Li \and Yang Liu \and Hao Song 
}

\institute{Guanmin Guo$^{*}$ \at \email{gmguo$\_$xjtukgd@yeah.net}   \\
          Ruihu Li \at \email{llzsy110@126.com}  \\
          Yang Liu \at \email{liu$\_$yang10@163.com}   \\
          Hao Song \at \email{songhao$\_$kgd@163.com}   \\
\at Fundamentals Department, Air Force Engineering University,
 Xi'an, Shaanxi 710051,  P. R. China}

\date{Received: date / Accepted: date}

\maketitle

\begin{abstract}
Self-dual MDS and NMDS codes over finite fields are linear codes
with significant combinatorial and cryptographic applications. In
this paper, firstly, we investigate the duality properties of
generalized twisted Reed-Solomon (abbreviated GTRS) codes in some
special cases. In what follows, a new systematic approach is
proposed to draw Hermitian self-dual (+)-GTRS codes. The necessary
and sufficient conditions of a Hermitian self-dual (+)-GTRS code are
presented. With this method, several classes of Hermitian self-dual
MDS and NMDS codes are constructed.

\keywords{Hermitian self-dual \and  generalized twisted Reed-Solomon
codes \and  MDS codes \and   NMDS codes}
\end{abstract}

\section{Introduction}

Maximum distance separable (MDS) codes are optimal because they
attain the maximal achievable minimum distance $d=n-k+1$ of length
$n$ and dimension $k$, which have the largest error-correcting
capability for given a code rate. The most famous family of MDS
codes is (extended) generalized Reed-Solomon (for short GRS and
EGRS) codes. There are, of course, other non-Reed-Solomon type MDS
codes \cite{Roth1}. Near MDS (i.e. NMDS) codes are introduced in
\cite{Dodunekov} by slightly weakening the restrictive conditions in
the definition of MDS codes, which are closely connected to
interesting objects in finite geometry and have applications in
combinatorics \cite{Dodunekov,Landjev} and secret sharing scheme
\cite{Zhou}. Similarly, because of their special algebraic
structure, self-dual codes are another family of linear codes worth
studying, and have important applications in cryptographic protocols
\cite{Dougherty,Massey}. For those reasons, constructing (Hermitian)
self-dual MDS and NMDS codes is thus becoming a significant research
topic in the theory of classical error-correcting codes. Analogous
with the construction of (Hermitian) self-dual MDS codes, it is also
challenging to determine the existence of a (Hermitian) self-dual
NMDS code.

In recent years, researchers are trying to use different techniques
to focus on investigating Euclidean and Hermitian self-dual MDS
codes, especially for Euclidean case, via building-up construction
method \cite{Kim,Gulliver}, and constacyclic codes
\cite{Guenda,Tong}, Glynn codes \cite{Baicheva}, rational function
fields \cite{Sok}. Especially recently, many Euclidean self-dual MDS
codes have been presented by utilizing GRS codes
\cite{Jin,Fang,Fang1,Fang4,Zhang2}. In \cite{Gulliver}, Gulliver
\emph{et al.} also construct Euclidean self-dual NMDS codes of
length $n=q-1$ ($q$ is power of odd prime) derived from Reed-Solomon
(i.e. RS) codes. In \cite{Kotsireas}, some self-dual NMDS codes with
length $n \leq 16$ were constructed over some small prime fields.
Jin and Kan \cite{Jin1} make use of properties of elliptic curves to
construct some self-dual NMDS codes. Consequently, constructing
self-dual NMDS codes remains an open problem for a large range of
parameters. As far as we know, however, there are few research
results on Hermitian self-dual MDS and NMDS codes, for a few
results, see \cite{Guo,Niu}.


In 2017, enlighten by the construction of twisted Gabidulin codes
\cite{Sheekey} in rank metric, Beelen \emph{et al.} \cite{Beelen1}
introduce a new family of linear evaluation codes in Hamming metric:
twisted Reed-Solomon (i.e. TRS) codes. The idea of TRS codes is
based on RS codes, by adding further monomials, so called ``twist",
and selecting the evaluation points appropriately. Afterwards,
Beelen \emph{et al.} \cite{Beelen2} also propose the generalization
of the single-twist Reed-Solomon codes in \cite{Beelen1} to the
multi-twist composition. TRS codes are also shown to be largely
distinct from GRS codes, which have much larger Schur squares
dimension than a GRS code with the same parameters. Meanwhile, a
subfamily of TRS codes are proposed as an alternative to Goppa codes
for the McEliece cryptosystem \cite{Beelen2,Lavauzelle}, which is a
public-key cryptosystem and one of the candidates for post-quantum
cryptography, resulting in a potential reduction of key sizes. We
call the extension of TRS codes by generalized TRS (i.e. GTRS)
codes. For other recent studies on GTRS codes, please refer to
\cite{Huang,Wu,WHL2021}. In general, TRS codes are not MDS,
nevertheless certain subclasses may be MDS or NMDS which are
constructed by a suitable choice of the evaluation points and twist
coefficients. What's more famous is that (+)-twisted Reed-Solomon
codes \cite{Beelen1}, which is called (+)-TRS codes for simplicity.
In \cite{Huang}, Huang \emph{et al.} represent the form of check
matrix of (+)-GTRS codes.


In this paper, we firstly prove that GTRS codes are also closed
under Euclidean duality if we choose evaluation points which form a
multiplicative group. In the following, we present the necessary and
sufficient conditions of a (+)-GTRS code is Hermitian self-dual and
give a new efficient construction method for self-dual (+)-GTRS
codes with respect to the Hermitian inner product. By applying the
new method, we draw several classes of Hermitian self-dual MDS and
NMDS codes, respectively.

The remainder of this paper is organized as follows. Basic notations
and results about GTRS codes and NMDS codes are provided in Section
II. The main contributions are presented in Section III. Some final
remarks and hints for future works conclude the paper in Section IV.

\section{Preliminaries}

In this section, we recall some definitions and basic theory of
Hermitian self-dual codes, GTRS codes, and NMDS codes.

\subsection{Hermitian self-dual codes}

Let $q$ be a prime power and $\mathbb{F}_{q}$ be the finite field
with $q$ elements. Assume that $n$ and $q$ are coprime, that is
gcd$(n,q)=1, \mathbb{F}_{q^{*}}=\mathbb{F}_{q}\backslash \{0\}$. Let
$\mathbb{F}_{q}^{n}$ denote the vector space of all $n$-tuples over
the finite field $\mathbb{F}_{q}$. If $\mathcal{C}$ is a
$k$-dimensional subspace of $\mathbb{F}_{q}^{n}$, then $\mathcal{C}$
will be called an $[n, k]$ linear code over $\mathbb{F}_{q}.$ The
linear code $\mathcal{C}$ has $q^{k}$ codewords.

Let $\mathbf{x}=(x_{1},x_{2},\ldots,x_{n})$,
$\mathbf{y}=(y_{1},y_{2},\ldots,y_{n})\in \mathbb{F}_{q^2}^{n}$,
here we review that the Euclidean inner product of vectors
$\mathbf{x}, \mathbf{y}$ is
\begin{equation}
\label{eqn_example} \langle\mathbf{x},
\mathbf{y}\rangle_{E}=\sum_{i=1}^{n} x_{i} y_{i}.
\end{equation}

The Euclidean dual code of $\mathcal{C}$ is defined as
\begin{equation}
\label{eqn_example1.1} \mathcal{C}^{\perp_{E}}=\{\mathbf{x} \mid
\mathbf{x} \in \mathbb{F}_{q^2}^n,
\left\langle\mathbf{x},\mathbf{y}\right\rangle_{E}=0, \text { for
all } \mathbf{y} \in\mathcal{C} \}.
\end{equation}

It is always useful to consider another inner product, called the
Hermitian inner product.
\begin{equation}
\label{eqn_example} \langle\mathbf{x},
\mathbf{y}\rangle_{H}=\sum_{i=1}^{n} x_{i} y_{i}^{q}.
\end{equation}

Analogous to (\ref{eqn_example1.1}), we can define the Hermitian
dual of $\mathcal{C}$ as follows by using this inner product.
\begin{equation}
\label{eqn_example} \mathcal{C}^{\perp_{H}}=\{\mathbf{x} \mid
\mathbf{x} \in \mathbb{F}_{q^2}^n,
\left\langle\mathbf{x},\mathbf{y}\right\rangle_{H}=0, \text { for
all } \mathbf{y} \in\mathcal{C} \}.
\end{equation}

Namely, $\mathcal{C}^{\perp_{H}}$ is the orthogonal subspace to
$\mathcal{C}$, with respect to the Hermitian inner product. We also
have Hermitian self-orthogonality and Hermitian self-duality. If
$\mathcal{C} \subseteq \mathcal{C}^{\perp_{H}}$, then
$\mathcal{C}^{\perp_{H}}$ is Hermitian self-orthogonal.
Particularly, if $\mathcal{C}^{\perp_{H}} = \mathcal{C}$, then
$\mathcal{C}$ is Hermitian self-dual.


\subsection{GTRS codes and NMDS codes}

The GTRS codes are formally defined as follows, for more details we
refer to \cite{Beelen1,Beelen2}.

\begin{definition} Let $n, k, \ell \in \mathbb{N}$ be positive integers, where
$k<n$, $\ell \leq n-k$. Choose a twist vector $\boldsymbol{t}=
\left(t_{1}, t_{2}, \ldots, t_{\ell}\right) \in\{1, \ldots,
n-k\}^{\ell}$ such that the $t_{i}(1 \leq i \leq \ell)$ are
distinct, and a hook vector $\boldsymbol{h} =\left(h_{1}, h_{2},
\ldots, h_{\ell}\right) \in\{0, \ldots, k-1\}^{\ell}$ such that the
$h_{i}(1 \leq i \leq \ell)$ are also distinct. Set
$\boldsymbol{\eta} =\left(\eta_{1}, \eta_{2}, \ldots,
\eta_{\ell}\right) \in (\mathbb{F}_{q}^{\ast})^{\ell}$. The set of
$[k, \boldsymbol{t}, \boldsymbol{h}, \boldsymbol{\eta}]$-twisted
polynomials over $\mathbb{F}_{q}$ is defined by
\begin{equation}
\mathcal{P}_{k, n}[\boldsymbol{t}, \boldsymbol{h},
\boldsymbol{\eta}]=\left\{f=\sum_{i=0}^{k-1} f_{i}
x^{i}+\sum_{j=1}^{\ell} \eta_{j} f_{h_{j}} x^{k-1+t_{j}}: f_{i} \in
\mathbb{F}_{q}\right\}.
\end{equation}
\end{definition}

\begin{definition} Let $\boldsymbol{\alpha}=\left(\alpha_{1},
\alpha_{2}, \ldots, \alpha_{n}\right) \in \mathbb{F}_{q}^{n}$ be
pairwise distinct, $\boldsymbol{v}=(v_{1}, v_{2}, $ $\ldots, v_{n})
\in (\mathbb{F}_{q}^{\ast})^{n}$ and $1 \leq k \leq n.$ Let
$\boldsymbol{t}, \boldsymbol{h}, \boldsymbol{\eta}$ and
$\mathcal{P}_{k, n}[\boldsymbol{t}, \boldsymbol{h},
\boldsymbol{\eta}]$ be defined as above. The $[\boldsymbol{\alpha},
\boldsymbol{v}, \boldsymbol{t}, \boldsymbol{h}, \boldsymbol{\eta}]$
-GTRS code of length $n$ and dimension $k$ is defined by
\begin{equation}
GTRS_{k,n}[\boldsymbol{\alpha}, \boldsymbol{v}, \boldsymbol{t},
\boldsymbol{h}, \boldsymbol{\eta}]:=\{[v_{1}f(\alpha_{1}),
v_{2}f(\alpha_{2}),\ldots, v_{n}f(\alpha_{n})]:f \in \mathcal{P}_{k,
n}[\boldsymbol{t}, \boldsymbol{h}, \boldsymbol{\eta}]\}.
\end{equation}
\end{definition}


The elements $\alpha_{1}, \alpha_{2}, \ldots, \alpha_{n}$ are called
the \textit{code locators (evaluation points)} of $GTRS_{k,n}$
$[\boldsymbol{\alpha}, \boldsymbol{v}, \boldsymbol{t},
\boldsymbol{h}, \boldsymbol{\eta}]$, and the elements $v_{1}, v_{2},
\ldots, v_{n}$ are called the \textit{column multipliers}. The set
of twisted polynomials $\mathcal{P}_{k, n}[\boldsymbol{t},
\boldsymbol{h}, \boldsymbol{\eta}] \subseteq \mathbb{F}_{q}[x]$
forms a $k$-dimensional $\mathbb{F}_{q}$-linear subspace, so a GTRS
code is linear code.

Let us recall the definition of NMDS codes as follows.

\begin{definition} (\cite{Dodunekov}) A linear code with
parameters of the form $[n,k,n - k]$ is said to be almost MDS (i.e.
AMDS). Particularly, An AMDS code is an NMDS code if the dual code
is also an AMDS code.
\end{definition}

\section{Main Results}

\subsection{Euclidean dual of GTRS codes}

It is known that the dual code of a GRS code is also a GRS code. In
contrast to GRS codes, GTRS also do not generally seem to be closed
under duality. However, if we choose evaluation points which form a
multiplicative group, this yields to the following results.

Firstly, denote the reversal matrix $\boldsymbol{J}_{k} \in
\mathbb{F}_{q}^{k \times k}$ by the square matrix
\begin{equation}
\label{eqn_example} \boldsymbol{J}_{k}=\left(\begin{array}{lll}
& & 1 \\
&  \iddots & \\
1 & &
\end{array}\right).
\end{equation}

We denote by $\boldsymbol{V}_{n}(\boldsymbol{\alpha})$ the $n \times
n$ Vandermonde matrix over $\boldsymbol{\alpha},$ and
$\boldsymbol{\Lambda}$ is the diagonal matrix
$\operatorname{diag}\left(v_{1}, v_{2}, \ldots, v_{n}\right),$ where
\begin{equation}
\label{eqn_example}
\boldsymbol{V}_{n}(\boldsymbol{\alpha})=\left(\begin{array}{cccc}
1 & 1 & \ldots & 1 \\
\alpha_{1} & \alpha_{2} & \ldots & \alpha_{n} \\
\vdots & \vdots & \ddots & \vdots \\
\alpha_{1}^{n-1} & \alpha_{2}^{n-1} & \ldots & \alpha_{n}^{n-1}
\end{array}\right).
\end{equation}

\begin{theorem} Let $\mathcal{C}$ be an $[n, k]$ linear code
with a generator matrix of the form
\begin{equation}
\label{eqn_example} \boldsymbol{G}=\left[\begin{array}{ll}
\boldsymbol{I} \mid \boldsymbol{L}
\end{array}\right] \cdot
(\boldsymbol{V}_{n}(\boldsymbol{\alpha})\boldsymbol{\Lambda}),
\end{equation}
where $\boldsymbol{I} \in \mathbb{F}_{q}^{k \times k}$ is the
identity matrix, $\boldsymbol{L} \in \mathbb{F}_{q}^{k \times
(n-k)},$ and the entries of $\boldsymbol{\alpha} \in
\mathbb{F}_{q}^{n}$ are distinct and form a multiplicative group.
Then the Euclidean dual code $\mathcal{C}^{\perp_{E}}$ has generator
matrix with the form
\begin{equation}
\label{eqn_example} \boldsymbol{H}=[\boldsymbol{I} \mid
\boldsymbol{J}_{n-k} (-\boldsymbol{L}^{T}) \boldsymbol{J}_{k}] \cdot
\boldsymbol{V}_{n}(\boldsymbol{\alpha})
\operatorname{diag}(\boldsymbol{\alpha} /
n)\boldsymbol{\Lambda}^{-1}.
\end{equation}
\end{theorem}

\begin{proof} 


%
Since the entries of $\boldsymbol{\alpha}$ form a multiplicative
group, we have $\alpha_{i}^{n}=1$,  $1 \leq i \leq n$ and by
\cite{Beelen2}, we obtain
\begin{equation}
\label{eqn_example} (\boldsymbol{V}^{T})^{-1}=\boldsymbol{J} \cdot
\boldsymbol{V} \cdot \operatorname{diag}(\boldsymbol{\alpha} / n).
\end{equation}

Since $\boldsymbol{H}$ has rank $n-k$ so left is to show
$\boldsymbol{G} \cdot \boldsymbol{H}^{\mathrm{T}}=\mathbf{0}$. Note
that
{\normalsize
\begin{eqnarray*}
 && \boldsymbol{G} \cdot \boldsymbol{H}^{T} \\
&=& [\boldsymbol{I} \mid \boldsymbol{L}] (
\boldsymbol{V}\boldsymbol{\Lambda}) \cdot ([\boldsymbol{I} \mid
\boldsymbol{J}_{n-k} (-\boldsymbol{L}^{T}) \boldsymbol{J}_{k}] \cdot
\boldsymbol{V}
\operatorname{diag}(\boldsymbol{\alpha} / n)\boldsymbol{\Lambda}^{-1})^{T} \\
&=& [\boldsymbol{I} \mid \boldsymbol{L}]
(\boldsymbol{V}\boldsymbol{\Lambda}) \cdot
(\boldsymbol{J}_{n-k}[-\boldsymbol{L}^{\mathrm{T}} \mid
\boldsymbol{I}]  \boldsymbol{J}_{n} \cdot
\boldsymbol{V} \operatorname{diag}(\boldsymbol{\alpha} / n)\boldsymbol{\Lambda}^{-1})^{T}\\
&=& [\boldsymbol{I} \mid \boldsymbol{L}]
(\boldsymbol{V}\boldsymbol{\Lambda}) \cdot
(\boldsymbol{J}_{n-k}[-\boldsymbol{L}^{\mathrm{T}} \mid
\boldsymbol{I}] (\boldsymbol{V}^{-1})^{\mathrm{T}} \boldsymbol{\Lambda}^{-1})^{T}\\
&=& [\boldsymbol{I} \mid \boldsymbol{L}][\begin{array}{c}
-\boldsymbol{L} \\
\boldsymbol{I}
\end{array}] \boldsymbol{J}_{n-k} \\
&=& \mathbf{0}.
\end{eqnarray*}}
\noindent So it is a parity-check matrix of $\mathcal{C},$ and thus,
a generator matrix of the dual code.
\end{proof}

Theorem 1 implies the following duality statement for GTRS codes
with evaluation points forming a multiplicative group, analogy to
TRS codes in \cite{Beelen2}.

\begin{theorem} Let $n, k, \boldsymbol{\alpha}, \boldsymbol{v},
\boldsymbol{t}, \boldsymbol{h}, \boldsymbol{\eta}$ be chosen as in
Definition 2 such that the entries of $\boldsymbol{\alpha}$ form a
multiplicative subgroup of $\mathbb{F}_{q}^{\ast}$. Then $G T R
S_{k, n}[\boldsymbol{\alpha}, \boldsymbol{v}, \boldsymbol{t},
\boldsymbol{h}, \boldsymbol{\eta}]^{\perp_{E}}$ twisted code is
equivalent to a $G T R S_{n-k, n}[\boldsymbol{\alpha},
\boldsymbol{v}^{-1}, k-\boldsymbol{h},
n-k-\boldsymbol{t},-\boldsymbol{\eta}]$ -twisted code.
\end{theorem}

\begin{proof} By definition, we claim that a generator matrix of $G T
R S_{k, n}[\boldsymbol{\alpha}, \boldsymbol{v}, \boldsymbol{t},
\boldsymbol{h}, \boldsymbol{\eta}]$ is given by
$\boldsymbol{G}=[\boldsymbol{I} \mid \boldsymbol{L}] \cdot
(\boldsymbol{V} \boldsymbol{\Lambda}),$ where the entries of
$\boldsymbol{L} \in \mathbb{F}_{q}^{k \times (n-k)}$ are of the form
\begin{equation}
\label{eqn_example} \boldsymbol{L}_{i j}=\left\{\begin{array}{ll}
\eta_{\mu}, & \text { if }(i, j)=\left(h_{\mu}+1, t_{\mu}\right), \\
0, & \text { else }
\end{array}\right.
\end{equation}

With the analysis as Theorem 1, a parity check matrix for $G T R
S_{k, n}[\boldsymbol{\alpha}, \boldsymbol{v}, \boldsymbol{t},
\boldsymbol{h}, \boldsymbol{\eta}]$ is:
\begin{equation}
\label{eqn_example} \boldsymbol{H}=[\boldsymbol{I} \mid
\boldsymbol{J}_{n-k} (-\boldsymbol{L}^{T}) \boldsymbol{J}_{k}] \cdot
\boldsymbol{V}_{n}(\boldsymbol{\alpha})
\operatorname{diag}(\boldsymbol{\alpha} /
n)\boldsymbol{\Lambda}^{-1}.
\end{equation}

Hence it is equivalent to a code $\mathcal{C}^{\prime}$ generated by
$\left[\boldsymbol{I} \mid-\boldsymbol{J}_{n-k} \boldsymbol{L}^{T}
\boldsymbol{J}_{k}\right] \cdot
\boldsymbol{V}_{n}(\boldsymbol{\alpha})\boldsymbol{\Lambda}^{-1}$.
As we already know, the entries of $-\boldsymbol{J}_{n-k}
\boldsymbol{L}^{\mathrm{T}} \boldsymbol{J}_{k}$ are of the form
\begin{equation}
\label{eqn_example} (-\boldsymbol{J}_{n-k}
\boldsymbol{L}^{\mathrm{T}}
\boldsymbol{J}_{k})_{i,j}=\left\{\begin{array}{ll}-\eta_{\mu}, & (i,
j)=\left(n-k-t_{\mu}+1, k-h_{\mu}\right), \\ 0, & \text { else.
}\end{array}\right.
\end{equation}

In other words, a twist $x^{h_{\mu}}+\eta_{\mu} x^{k-1+t_{\mu}}$
becomes the twist $x^{n-k-t_{\mu}}+\left(-\eta_{\mu}\right)
x^{n-k-1+\left(k-h_{\mu}\right)}$ in the dual code. Therefore the
code $\mathcal{C}^{\prime}$ is a $[k-\boldsymbol{h},
n-k-\boldsymbol{t},-\boldsymbol{\eta}]$ -twisted code, which proves
the claim.
\end{proof}

\subsection{$(+)$ -generalized twisted Reed-Solomon codes}

Taking $l=1,$ $(t,h) = (1, k-1)$ in Definition 2, Beelen \emph{et
al.} obtain a family code as the $(+)$-twisted Reed-Solomon codes by
employing additive subgroups of $\mathbb{F}_{q}$. We denote
generalization of the class twisted code as $ G T R S_{k,
n}[\boldsymbol{\alpha}, \boldsymbol{v}, 1, k-1, \eta]$.

\begin{lemma} (\cite{Beelen1}) Let $k\leq n \leq q,$
$\boldsymbol{\alpha}=\left(\alpha_{1}, \alpha_{2}, \ldots,
\alpha_{n}\right) \in \mathbb{F}_{q}^{n}$ be pairwise distinct,
$\boldsymbol{v}=\left(v_{1}, v_{2}, \ldots, v_{n}\right) \in
(\mathbb{F}_{q}^{\ast})^{n}$, and $\eta \in \mathbb{F}_{q}^{*}$.
Then the generalized twisted code $ G T R S_{k,
n}[\boldsymbol{\alpha}, \boldsymbol{v}, 1, $ $k-1, \eta]$ is
$\mathrm{MDS}$ if and only if
\begin{equation}
\label{eqn_example} \eta \sum_{i \in \mathcal{I}} \alpha_{i} \neq-1,
\quad \forall ~\mathcal{I} \subseteq\{1, \ldots, n\} \text { s.t.
}|\mathcal{I}|=k.
\end{equation}
\end{lemma}

Next, we present the sufficient and necessary conditions that
$(+)$-GTRS code is an NMDS code. It is easy to conclude from the
proof process of Lemma 1, so we omit the details.

\begin{lemma} Let $k,n,\boldsymbol{\alpha},\boldsymbol{v},\eta$
be chosen as above. Then $ G T R S_{k, n}[\boldsymbol{\alpha},
\boldsymbol{v}, 1, k-1, \eta]$ is NMDS if and only if
\begin{equation}
\label{eqn_example} \eta \sum_{i \in \mathcal{I}} \alpha_{i} = -1,
\quad \exists ~\mathcal{I} \subseteq\{1, \ldots, n\} \text { s.t.
}|\mathcal{I}|=k.
\end{equation}
\end{lemma}

\begin{remark} It can be drawn that the code $G T R
S_{k,n}[\boldsymbol{\alpha}, \boldsymbol{v}, 1, k-1, \eta]$ is MDS
if $-\eta^{-1}$ cannot be represented as the sum of any $k$
\emph{evaluation points}. Furthermore, $\forall   ~\eta \in
\mathbb{F}_{q}^{*}$, $G T R S_{k, n}[\boldsymbol{\alpha},
\boldsymbol{v}, 1, k-1, \eta]$ is either MDS or NMDS.
\end{remark}

\subsection{Hermitian self-dual $(+)$-GTRS codes}

From now on, we always assume that $\omega$ is a primitive element
of $\mathbb{F}_{q^{2}}$, that is
$\mathbb{F}_{q^{2}}^{*}=\langle\omega\rangle$, and label the
elements of $\mathbb{F}_{q}$ as $\mathbb{F}_{q} =\left\{a_{1},
a_{2}, \ldots, a_{q}\right\}$.

Meanwhile, we also always denote $\boldsymbol{u}=\left(u_{1}, u_{2},
\ldots, u_{n}\right)$, where
\begin{equation}
\label{eqn_example} u_{i}:=\prod_{1 \leq j \leq n, j \neq
i}\left(\alpha_{i}-\alpha_{j}\right)^{-1}, ~1\leq i \leq n,
\end{equation}

\noindent and
\begin{equation}
\label{eqn_example} a=\sum_{i=1}^{n} \alpha_{i}.
\end{equation}

Next according to the check matrix of $G T R S_{k,
n}[\boldsymbol{\alpha}, \boldsymbol{v}, 1, k-1, \eta]$ in
\cite{Huang}, we present the following lemma.

\begin{lemma} Let $k\leq n \leq q^2,$
$\boldsymbol{\alpha}=\left(\alpha_{1}, \alpha_{2}, \ldots,
\alpha_{n}\right) \in \mathbb{F}_{q^2}^{n}$ be pairwise distinct,
$\boldsymbol{v}=\left(v_{1}, v_{2}, \ldots, v_{n}\right)$ $ \in
(\mathbb{F}_{q^2}^{\ast})^{n}$, and $\eta \in \mathbb{F}_{q^2}^{*}$.
Then the Euclidean dual of twisted code $G T R S_{k,n}$
$[\boldsymbol{\alpha}, \boldsymbol{v}, 1, k-1, \eta](\eta
\neq-a^{-1})$ is represented as follows.
\begin{eqnarray}
& & G T R S_{k, n}^{\perp_{E}}[\boldsymbol{\alpha}, \boldsymbol{v},
1, k-1, \eta]\nonumber\\
&=& G T R S_{n-k, n}[\boldsymbol{\alpha},
\boldsymbol{u}\boldsymbol{v}^{-1}, 1, n-k-1, -\frac{\eta}{1+a \eta}]
\nonumber.
\end{eqnarray}
\end{lemma}


\begin{remark} In Theorem 2, suppose that
$\boldsymbol{\alpha}$ form a multiplicative subgroup of
$\mathbb{F}_{q^{2}}$, then $a=\sum_{i=1}^{n} \alpha_{i}=0$, and set
$l=1,$ $(t, h) = (1, k - 1)$, then $ G T R S_{k,
n}^{\perp_{E}}[\boldsymbol{\alpha}, \boldsymbol{v}, 1, k-1, \eta]= G
T R S_{n-k, n}$ $[\boldsymbol{\alpha},
\boldsymbol{u}\boldsymbol{v}^{-1}, 1, n-k-1, -\eta].$ Thus the
result of Lemma 3 is a special case of Theorem 2 and vice versa.
\end{remark}

According to Lemma 3, we obtain the corollary as follows.

\begin{cor} Let $\mathbf{1}$ be all-one word of length $n.$ Then the Euclidean
dual code of $G T R S_{k, n}[\boldsymbol{\alpha}, \boldsymbol{1},
1,$ $ k-1, \eta](\eta \neq-a^{-1})$ is
\begin{eqnarray}
& & G T R S_{k,n}^{\perp_{E}}[\boldsymbol{\alpha}, \boldsymbol{1}, 1, k-1, \eta]\nonumber\\
&=& G T R S_{n-k, n}[\boldsymbol{\alpha},
\boldsymbol{u}, 1, n-k-1, -\frac{\eta}{1+a \eta}] \nonumber \\
&=& \{(u_{1} g(\alpha_{1}), \ldots, u_{n} g(\alpha_{n}))|g(x) \in
\mathbb{F}_{q^{2}}[x]\} \nonumber,
\end{eqnarray}
%
%
%
%
%
%
\noindent where $g(x)=\sum_{i=0}^{n-k-2}
g_{i}x^{i}+g_{n-k-1}(x^{n-k-1} -\frac{\eta}{1+a \eta} x^{n-k})$,
$g_{i} \in \mathbb{F}_{q^2}, 0\leq i \leq n-k-1$ with $ g_{n-k-1}
\neq 0$.
\end{cor}

In the following, we show that the necessary and sufficient
conditions for (+)-GTRS codes being Hermitian self-dual.

\begin{theorem} Keep the above notations, let $n=2k$, then $ G T R
S_{k, n}[\boldsymbol{\alpha}, \boldsymbol{v}, 1, k-1, \eta](\eta
\neq-a^{-1})$ over $\mathbb{F}_{q^{2}}$ is Hermitian self-dual if
and only if there exists a polynomial $g(x)=\sum_{i=0}^{k-2}
g_{i}x^{i}+g_{k-1}(x^{k-1}-\frac{\eta}{1+a \eta} x^{k})$, $g_{i} \in
\mathbb{F}_{q^2}, 0\leq i \leq k-1$ with $ g_{k-1} \neq 0$ such that
%
\begin{equation}
\label{eqn_example} v_{i}^{q+1}f^{q}(\alpha_{i})=u_{i}g(\alpha_{i}),
1\leq i \leq n.
\end{equation}
\end{theorem}

\begin{proof} Note that $G T R S_{k, n}[\boldsymbol{\alpha},
\boldsymbol{v}, 1, k-1, \eta]$ has a generator matrix given by
$G_{k}(\boldsymbol{\alpha}, \boldsymbol{v},\eta)$. Clearly, we have
$G_{k}(\boldsymbol{\alpha},
\boldsymbol{v},\eta)=G_{k}(\boldsymbol{\alpha}, \boldsymbol{1},\eta)
\Lambda$, where

$$
G_{k}(\boldsymbol{\alpha},
\boldsymbol{1},\eta)=\left(\begin{array}{cccc}
1 & 1 & \cdots & 1 \\
\alpha_{1} & \alpha_{2} & \cdots & \alpha_{n} \\
\vdots & \vdots & \ddots & \vdots \\
\alpha_{1}^{k-2} & \alpha_{2}^{k-2} & \cdots & \alpha_{n}^{k-2}\\
\alpha_{1}^{k-1}+\eta \alpha_{1}^{k} & \alpha_{2}^{k-1}+\eta
\alpha_{2}^{k} & \cdots &
\alpha_{n}^{k-1}+\eta \alpha_{n}^{k} \\
\end{array}\right),
$$

\noindent and $\Lambda$ is the diagonal matrix
$\operatorname{diag}\left(v_{1}, v_{2}, \ldots, v_{n}\right)$.
It follows that $G T R S_{\frac{n}{2}, n}[\boldsymbol{\alpha},
\boldsymbol{v}, 1, k-1, \eta]$ over $\mathbb{F}_{q^{2}}$ is
Hermitian self-dual if and only if for any codeword
$\mathbf{c}=(v_{1} f(\alpha_{1}), $ $v_{2} f(\alpha_{2}), \ldots, $
$v_{n} f(\alpha_{n}))$ of $G T R S_{\frac{n}{2},
n}[\boldsymbol{\alpha}, \boldsymbol{v}, 1, k-1, \eta]$,
\begin{eqnarray}
& & {\mathbf{c}^{q}} \cdot G_{\frac{n}{2}}(\boldsymbol{\alpha},
\boldsymbol{v},\eta)^{T} \nonumber\\
&=& {\mathbf{c}^{q}} \cdot (G_{\frac{n}{2}}(\boldsymbol{\alpha},
\boldsymbol{1},\eta) \Lambda) ^{T} \nonumber \\
&=& (v_{1}^{q+1} f^{q}(\alpha_{1}), \ldots, v_{n}^{q+1}
f^{q}(\alpha_{n})) \cdot  G_{\frac{n}{2}}(\boldsymbol{\alpha},
\boldsymbol{1},\eta)^{T} \nonumber\\
&=& \mathbf{0}\nonumber \\
&\Leftrightarrow & (v_{1}^{q+1} f^{q}(\alpha_{1}), \ldots,
v_{n}^{q+1} f^{q}(\alpha_{n})) \in G T R S_{\frac{n}{2},
n}^{\perp_{E}}[\boldsymbol{\alpha}, \boldsymbol{1}, 1, k-1,
\eta]\nonumber.
\end{eqnarray}


\noindent Recall that the Euclidean dual of $G T R S_{\frac{n}{2},
n}[\boldsymbol{\alpha}, \boldsymbol{1}, 1, k-1, \eta]$ is $G T R
S_{\frac{n}{2}, n}[\boldsymbol{\alpha}, \boldsymbol{u}, 1, k-1,
-\frac{\eta}{1+a \eta}]$, now the desired result follows immediately
from Corollary 3.
\end{proof}

\subsection{Hermitian self-dual MDS and NMDS codes}

In this section, we mainly present our contribution to construct
several classes of Hermitian self-dual MDS and NMDS codes. To do
that, we consider the Hermitian self-dual (+)-GTRS codes in Theorem
4. We first give the following basic lemmas from \cite{Mullen}.

\begin{lemma} If $\omega$ is a primitive element of
$\mathbb{F}_{q^{2}}$, then there exists a $\xi \in
\mathbb{F}_{q^{2}}$ such that $\omega^q+\omega=\xi^{q+1}$, that is
$\omega^q+\omega \in \mathbb{F}_{q}$.
\end{lemma}

\begin{proof} Since $(\omega^q+\omega)^q=\omega^{q^2}+\omega^q=\omega+\omega^q$, that is
$(\omega^q+\omega)^{q-1}=1$, then it is a straight-forward fact that
$\omega^q+\omega \in \mathbb{F}_{q}$.
\end{proof}


\begin{lemma} The equation $ \zeta^{q}+\zeta^{q-1}+1=0$ with
regard to $\zeta$ has $q$ distinct nonzero roots over the finite
field $\mathbb{F}_{q^{2}}$.
\end{lemma}

Next, we present our discussions according to two classes different
values of \emph{code locators} $\boldsymbol{\alpha}$.

(I) Fix $\beta \in \mathbb{F}_{q^{2}} \backslash \mathbb{F}_{q}$. $
\forall  ~1 \leq l \leq q $, set
\begin{equation}
\label{eqn_example} A_{l}=a_{l} \beta+\mathbb{F}_{q}:=\{{a_{l}
\beta+x: x \in \mathbb{F}_{q}}\}.
\end{equation}



In general, here we always set $\beta=\omega$.

\begin{theorem} Let $q$ be a prime power, $n=2k, n\leq q$,
$\boldsymbol{\alpha}=\left(\alpha_{1}, \alpha_{2}, \ldots,
\alpha_{n}\right) \in A_{l}^{n}$, where $\alpha_{1}, \alpha_{2},
\ldots, \alpha_{n}$ are distinct elements. If $a=0$ and $q = 2^s$
are not met at the same time, then there exists a vector
$\boldsymbol{v}=\left(v_{1}, v_{2}, \ldots, v_{n}\right) \in
(\mathbb{F}_{q^{2}}^{*})^{n}$, and $\eta \in
\mathbb{F}_{q^{2}}^{\ast}$ such that $G T R S_{\frac{n}{2},
n}[\boldsymbol{\alpha}, \boldsymbol{v}, 1, k-1, \eta]$ is an
$\left[n, \frac{n}{2}, \frac{n}{2}+1\right]$ Hermitian self-dual
GTRS code over $\mathbb{F}_{q^{2}}$.
\end{theorem}

\begin{proof} As can be seen, $|A_{l}|=q$. Let
$\boldsymbol{\alpha}=\left(\alpha_{1}, \alpha_{2}, \ldots,
\alpha_{n}\right) \in A_{l}^{n}$, then it is a straight-forward fact
that
\begin{equation}
\label{eqn_example} u_{i} = \prod_{1 \leq j \leq n, j \neq
i}\left(x_{i}-x_{j}\right)^{-1}.
\end{equation}


%

It is obvious that $u_{i} \in \mathbb{F}_{q}^{\ast}$, thus there
exists $v_{i} \in \mathbb{F}_{q^{2}}^{*}$ such that
$v_{i}^{q+1}=u_{i}.$ Set $\mathbf{v}=\left(v_{1}, v_{2}, \ldots,
v_{n}\right)$.

Let $\omega^q+\omega=\xi^{q+1}$, then
\begin{eqnarray}
 \alpha_{i}^{q}
&=&(a_{l} \omega+x_{i})^{q} \nonumber\\
&=& a_{l}^{q} \omega^{q}+x_{i}^{q} \nonumber \\
&=& a_{l} (\xi^{q+1}-\omega)+x_{i} \nonumber \\
&=&  (a_{l}\omega+x_{i})+(\xi^{q+1}-2 \omega)a_{l} \nonumber \\
&=& \alpha_{i}+(\xi^{q+1}-2 \omega)a_{l} \nonumber.
\end{eqnarray}

\noindent For all $f(x) \in \mathbb{F}_{q^{2}}[x]$ with form
$f(x)=\sum_{i=0}^{k-2} f_{i} x^{i}+ f_{k-1}(x^{k-1}+ \eta x^{k}),
f_{k-1}\neq 0$, we will discuss it in two ways.

(1) In the case of $a=0$ and $q \neq 2^s$, set $\eta^q=-\eta$, and
$h(x)=\sum_{i=0}^{k-2} f_{i}^{q} x^{i}+ f_{k-1}^{q}(x^{k-1}- \eta
x^{k})$. By $\alpha_{i}^{q}=\alpha_{i}+(\xi^{q+1}-2\omega)a_{l}$,
therefore
\begin{eqnarray}
 f^{q}(\alpha_{i})
&=& \sum_{j=0}^{k-2}f_{j}^{q}(\alpha_{i}^{q})^{j}+f_{k-1}^{q}((\alpha_{i}^{q})^{k-1}+\eta^{q}(\alpha_{i}^{q})^{k}) \nonumber\\
&=& h(\alpha_{i}+(\xi^{q+1}-2\omega)a_{l}). \nonumber
\end{eqnarray}

%

Set $g(x)=h(x+(\xi^{q+1}-2\omega)a_{l})$, then there exists
$g(x)=\sum_{i=0}^{k-2} g_{i}x^{i}+g_{k-1}(x^{k-1}-\eta x^{k}) \in
\mathbb{F}_{q^2}[x]$ with $ g_{k-1} \neq 0$ such that
$f^{q}(\alpha_{i})=g(\alpha_{i}), 1\leq i \leq n$. Therefore, there
exists a $g(x)$ such that $v_{i}^{q+1} f^{q}(\alpha_{i})=u_{i}
g(\alpha_{i}), 1\leq i \leq n.$ By Theorem 4, $G T R S_{\frac{n}{2},
n}[\boldsymbol{\alpha}, \boldsymbol{v}, 1, k-1, \eta]$ is a
Hermitian self-dual GTRS code.

(2) In the case of $a\neq 0$, set $\eta^q=\mu\eta$, $\mu \in
\mathbb{F}_{q^2} $ and $h(x)=\sum_{i=0}^{k-2} f_{i}^{q} x^{i}+
f_{k-1}^{q}(x^{k-1}+ \mu\eta x^{k})$. By
$\alpha_{i}^{q}=\alpha_{i}+(\xi^{q+1}-2\omega)a_{l}$, therefore
$f^{q}(\alpha_{i})=h(\alpha_{i}+(\xi^{q+1}-2\omega)a_{l}).$

Set $g(x)=h(x+(\xi^{q+1}-2\omega)a_{l})$, to make $g(x)$ has form
$g(x)=\sum_{i=0}^{k-2} g_{i}x^{i}+g_{k-1}(x^{k-1}-\frac{\eta}{1+a
\eta} x^{k}) \in \mathbb{F}_{q^2}[x]$ with $ g_{k-1} \neq 0$, by
analyzing the coefficient of $x^{k-1}$ and $x^{k}$ on both sides,
then
\begin{equation}
\label{eqn_example}
\frac{\mu\eta}{k(\xi^{q+1}-2\omega)a_{l}\mu\eta+1}=-\frac{\eta}{1+a\eta}.
\end{equation}

Combining with $\eta^q=\mu\eta$ and Equation (22), then
\begin{equation}
\label{eqn_example}
[k(\xi^{q+1}-2\omega)a_{l}+a]\eta^{q}+\eta^{q-1}+1=0.
\end{equation}


It is easy to prove that $A\triangleq k(\xi^{q+1}-2\omega)a_{l}+a=
\sum_{i=1}^{n} x_{i}+k \xi^{q+1}a_{l} \in \mathbb{F}_{q}$. Setting
$\zeta=A\eta$ transforms Equation (23) to
$\zeta^{q}+\zeta^{q-1}+A^{q-1}=0$, that is
$\zeta^{q}+\zeta^{q-1}+1=0$. By Lemma 5, Equation (23) has $q$
distinct nonzero roots in $\mathbb{F}_{q^{2}}$. Then there exists a
$g(x)$ such that $f^{q}(\alpha_{i})=g(\alpha_{i}), 1\leq i \leq n$.
Therefore, there exists a $g(x)$ such that $v_{i}^{q+1}
f^{q}(\alpha_{i})=u_{i} g(\alpha_{i}), 1\leq i \leq n.$ By Theorem
4, $G T R S_{\frac{n}{2}, n}[\boldsymbol{\alpha}, \boldsymbol{v}, 1,
k-1, \eta]$ is a Hermitian self-dual GTRS code, which proves the
claim.
\end{proof}

(II) Let $\beta_{m}=\omega^{m}, 1\leq m \leq q$, $ \forall 1 \leq l
\leq q$, denote
\begin{equation}
\label{eqn_example} A_{l,m}=a_{l} +\mathbb{F}_{q} \cdot \beta_{m}
:=\{{a_{l}+ \beta_{m} x: x \in \mathbb{F}_{q}}\}.
\end{equation}

\begin{theorem} Let $q$ be a prime power, $n=2k, n\leq q$,
$\boldsymbol{\alpha}=\left(\alpha_{1}, \alpha_{2}, \ldots,
\alpha_{n}\right) \in A_{l,m}^{n}$ with $\alpha_{1}, \alpha_{2},
\ldots, \alpha_{n}$ distinct elements. If $a=0$ and $q = 2^s$ are
not met at the same time, then there exists a vector
$\boldsymbol{v}=\left(v_{1}, v_{2}, \ldots, v_{n}\right) \in
(\mathbb{F}_{q^{2}}^{*})^{n}$, and $\eta \in
\mathbb{F}_{q^{2}}^{\ast}$ such that $G T R S_{\frac{n}{2},
n}[\boldsymbol{\alpha}, \boldsymbol{v}, 1, k-1, \eta]$ is an
$\left[n, \frac{n}{2}, \frac{n}{2}+1\right]$ Hermitian self-dual
GTRS code over $\mathbb{F}_{q^{2}}$.
\end{theorem}

\begin{proof} As can be seen, $|A_{l,m}|=q.$ Let
$\boldsymbol{\alpha}=(\alpha_{1}, \alpha_{2}, \ldots, \alpha_{n})\in
A_{l,m}^{n}$, then it can be shown that

\begin{equation}
\label{eqn_example} u_{i} = \beta_{m}^{-(n-1)}\prod_{1 \leq j \leq
n, j \neq i}\left(x_{i}-x_{j}\right)^{-1}.
\end{equation}

%
Let $\lambda=\beta_{m}^{n-1}=\omega^{m(n-1)} \in
\mathbb{F}_{q^{2}}^{*}$, thus there exists $v_{i} \in
\mathbb{F}_{q^{2}}^{*}$ such that $v_{i}^{q+1}=\lambda u_{i}.$ It
turns out that $\alpha_{i}^{q}=\beta_{m}^{q-1}
\alpha_{i}+(1-\beta_{m}^{q-1})a_{l}$. For all $f(x) \in
\mathbb{F}_{q^{2}}[x]$ with form $f(x)=\sum_{i=0}^{k-2} f_{i} x^{i}+
f_{k-1}(x^{k-1}+ \eta x^{k}), f_{k-1}\neq 0$, set
$h(x)=\sum_{i=0}^{k-2} f_{i}^{q} x^{i}+ f_{k-1}^{q}(x^{k-1}+ \mu\eta
x^{k})$, and $\eta^q=\mu\eta$. By $\alpha_{i}^{q}=\beta_{m}^{q-1}
\alpha_{i}+(1-\beta_{m}^{q-1})a_{l}$, then
\begin{eqnarray}
 f^{q}(\alpha_{i})
&=& \sum_{j=0}^{k-2}
f_{j}^{q}(\alpha_{i}^{q})^{j}+f_{k-1}^{q}((\alpha_{i}^{q})^{k-1}+\eta^{q}(\alpha_{i}^{q})^{k}) \nonumber\\
&=& h(\beta_{m}^{q-1} \alpha_{i}+(1-\beta_{m}^{q-1})a_{l}).
\nonumber
\end{eqnarray}

%

Set $g(x)=\lambda h(\beta_{m}^{q-1} x+(1-\beta_{m}^{q-1})a_{l})$, we
also consider the following two cases.

(1) In the case of $a=0$ and $q \neq 2^s$, to make $g(x)$ has the
form $g(x)=\sum_{i=0}^{k-2} g_{i}x^{i}+g_{k-1}(x^{k-1}-\eta x^{k})
\in \mathbb{F}_{q^2}[x]$ with $ g_{k-1} \neq 0$, then by considering
the coefficient of $x^{k-1}$ and $x^{k}$ on both sides, then
\begin{equation}
\label{eqn_example} \lambda\mu\eta(\beta_{m}^{q-1})^{k}=-\lambda
\eta (\beta_{m}^{q-1})^{k-1}.
\end{equation}
that is
\begin{equation}
\label{eqn_example} \mu \beta_{m}^{q-1}=-1.
\end{equation}

Combining with $\eta^q=\mu\eta$ and Equation (27), then
\begin{equation}
\label{eqn_example} \eta^{q-1}=-\beta_{m}^{-(q-1)}.
\end{equation}

Obviously, Equation (28) has $q-1$ distinct nonzero roots in
$\mathbb{F}_{q^{2}}$.

%
%
%

(2) In the case of $a \neq 0$, to make $g(x)$ has form
$g(x)=\sum_{i=0}^{k-2} g_{i}x^{i}+g_{k-1}(x^{k-1}-\frac{\eta}{1+a
\eta} x^{k}) \in \mathbb{F}_{q^2}[x]$ with $ g_{k-1} \neq 0$, by
analyzing the coefficient of $x^{k-1}$ and $x^{k}$ on both sides,
then
\begin{equation}
\label{eqn_example}
\frac{\mu\eta\beta_{m}^{q-1}}{1+k\mu\eta(1-\beta_{m}^{q-1})a_{l}}=-\frac{\eta}{1+a\eta}.
\end{equation}

Combining with $\eta^q=\mu\eta$ and Equation (29), then
\begin{equation}
\label{eqn_example}
[k(1-\beta_{m}^{q-1})a_{l}+a\beta_{m}^{q-1}]\eta^{q}+\beta_{m}^{q-1}\eta^{q-1}+1=0.
\end{equation}

Denoting $B\triangleq
\frac{k(1-\beta_{m}^{q-1})a_{l}+a\beta_{m}^{q-1}}{\beta_{m}^{q-1}}
$, and setting $\zeta=B\eta$ transforms Equation (30) to
$\zeta^{q}+\zeta^{q-1}+(B\beta_{m}^{-1})^{q-1}=0$, it is easy to
know $B\beta_{m}^{-1} \in \mathbb{F}_{q}$, that is
$\zeta^{q}+\zeta^{q-1}+1=0$. By Lemma 5, Equation (30) has $q$
distinct nonzero roots in $\mathbb{F}_{q^{2}}$.

From the above discussions, it follows that there exists a $g(x)$
such that $f^{q}(\alpha_{i})=\lambda^{-1} g(\alpha_{i}), 1\leq i
\leq n$. Therefore, there exists a $g(x)$ such that $v_{i}^{q+1}
f^{q}(\alpha_{i})=u_{i} g(\alpha_{i}), 1\leq i \leq n.$ By Theorem
4, the conclusion is established.
\end{proof}

\begin{remark} In the light of Theorem 2.5 in \cite{Huang},
suppose that $a=0$, then $G T R S_{\frac{n}{2},
n}[\boldsymbol{\alpha}, \boldsymbol{v}, $ $1, k-1, \eta]$  can not
be a Euclidean self-dual MDS code, however, it can be a Hermitian
self-dual MDS code.
\end{remark}

Building on Theorems 5 and 6, and by Lemmas 1 and 2, we derive two
striking conclusions.

\begin{cor} In Theorems 5 and 6, if $a=0$, then a
Hermitian self-dual GTRS code $G T R S_{\frac{n}{2},
n}[\boldsymbol{\alpha}, \boldsymbol{v}, 1, k-1, \eta]$ is a MDS code
over $\mathbb{F}_{q^{2}}$.
\end{cor}

\begin{cor} In Theorems 5 and 6, if $a\eta+2=0$,
then a Hermitian self-dual GTRS code $G T R S_{\frac{n}{2},
n}[\boldsymbol{\alpha}, \boldsymbol{v}, 1, k-1, \eta]$ is NMDS.
Otherwise, $G T R S_{\frac{n}{2}, n}[\boldsymbol{\alpha},
\boldsymbol{v}, 1, k-1, \eta]$ is a MDS code over
$\mathbb{F}_{q^{2}}$.
\end{cor}

\begin{example} To be more precise, let $q=7$, we present some
examples of Hermitian self-dual GTRS codes $G T R
S_{3,6}[\boldsymbol{\alpha}, \boldsymbol{v}, 1, 2, \eta]$ over
$\mathbb{F}_{7^2}$ in Table 1.
\end{example}

\begin{table}
\begin{center}
\caption{Some Hermitian self-dual $G T R
S_{3,6}[\boldsymbol{\alpha}, \boldsymbol{v}, 1, 2, \eta]$ with
parameters $[6,3,4]$ or $[6,3,3]$ over
$\mathbb{F}_{7^{2}}$.\label{tab1}} {\small
\begin{tabular}{lccccc}
\toprule
Class  & $a$ & Para. & $\boldsymbol{\alpha}$ & $\boldsymbol{v}$ & $\eta$   \\
\midrule (I) &$a=0$ & $[6,3,4]$  & $(1,2,3,4,5,6)$ &
$(\omega^4,1,\omega^{11},3, \omega^{9}, \omega)$ &
$\{\omega^4,\omega^{12},\omega^{20},\omega^{28}, \omega^{36}, \omega^{44}\}$  \\

(I) &$a \neq 0$ & $[6,3,4]$  &
$(\omega,\omega^2,\omega^5,\omega^{11},\omega^{31},\omega^{36})$ &
$(\omega^2,\omega^5,\omega^6,\omega^{10},\omega,\omega^{3})$ &
$\{\omega^{17},\omega^{23},\omega^{27},\omega^{38}, 5, \omega^{45}\}$  \\

(I) &$a \neq 0$ & $[6,3,3]$  &
$(\omega,\omega^2,\omega^5,\omega^{11},\omega^{31},\omega^{36})$ &
$(\omega^2,\omega^5,\omega^6,\omega^{10},\omega,\omega^{3})$ &
$\{\omega^{26}\}$  \\

(II) &$a = 0$ & $[6,3,4]$  &
$(\omega^4,\omega^{28},\omega^{20},\omega^{44},\omega^{12},\omega^{36})$
& $(\omega,\omega^{10},\omega^3,1,\omega^{2},\omega^{11})$ &
$\{1,2,3,4,5,6\}$  \\

(II) &$a \neq 0$ & $[6,3,4]$  &
$(\omega,\omega^{25},0,\omega^{17},\omega^{41},\omega^{9})$ &
$(\omega^2,1,\omega^{5},\omega^{3},\omega^{4},\omega)$ &
$\{3,\omega^{14},\omega^{17},\omega^{18},\omega^{29}, \omega^{36}\}$  \\

(II) &$a \neq 0$ & $[6,3,3]$  &
$(\omega,\omega^{25},0,\omega^{17},\omega^{41},\omega^{9})$ &
$(\omega^2,1,\omega^{5},\omega^{3},\omega^{4},\omega)$ &
$\{\omega^{31}\}$  \\
\bottomrule
\end{tabular}}
\end{center}
\end{table}

\begin{remark}
As a potential application in McEliece cryptosystem, GTRS codes play
an important role in reducing the public key size for a given
security level. In addition, according to \cite{Baldi2}, some
choices of the system parameters can avoid the mentioned attack,
e.g. using codes with rate $R \simeq \frac{1}{2}$. We know a
self-dual code have rate $R = \frac{1}{2}$. On the other hand,
people begin to construct cryptosystem by using variant codes of
original GRS and GTRS codes, It is worth noting that TRS codes are
also subcodes of GRS codes. We know the generator matrix and
dimension of subfield subcodes of GRS and GTRS codes are not
guaranteed and depends on the actual choice of code locators
$\mathbf{a}$, column multipliers $\mathbf{v}$ and variable $\eta$
\cite{Senger}. Our investigation on determining Hermitian self-dual
GTRS codes with pairs of $(\boldsymbol{\mathbf{a}},
\boldsymbol{\mathbf{v}},\eta)$, which are expected that these codes
and their subcodes can be used for constructing McEliece code-based
cryptosystems with resisting some more known structural attacks.
\end{remark}

\section{Conclusion and discussion}

In this paper, we mainly propose a systematical approach to
construct Hermitian self-dual (+)-GTRS codes for the first time.
Finally, we obtain several classes of $q^2$-ary Hermitian self-dual
MDS and NMDS codes derived from these GTRS codes.
Further, the techniques developed in this paper can be also applied
for these MDS codes in \cite{N2021} to obtain new Hermitian
self-dual MDS codes. Meanwhile it is also a worthy research topic to
construct Hermitian self-orthogonal (especially almost self-dual)
and Hermitian LCD MDS and NMDS codes through GTRS codes applying
this method.

\end{document}